\documentclass[floats,11pt,onecolumn]{article}
\usepackage{amsfonts}
\usepackage{amsmath}
\usepackage[utf8]{inputenc}  
\usepackage[T1]{fontenc}
\usepackage{calc}
\usepackage{subfigure}
\usepackage{graphicx} 
\usepackage{enumerate}
\usepackage{color}
\usepackage{amssymb}
\usepackage{amsthm}
\usepackage{mathtools}
\usepackage{hyperref}
\usepackage{listings}
\mathtoolsset{showonlyrefs}

\newcommand{\appref}[1]{\hyperref[#1]{{Appendix~\ref*{#1}}}}
\newcommand{\be}{\begin{eqnarray} \begin{aligned}}
\newcommand{\ee}{\end{aligned} \end{eqnarray} }
\newcommand{\benn}{\begin{eqnarray*} \begin{aligned}}
\newcommand{\eenn}{\end{aligned} \end{eqnarray*}}
\newcommand{\cancel}[1]{} 
\newcommand*{\cB}{\mathcal{B}}

\newcommand*{\cH}{\mathcal{H}}

\newcommand*{\tr}{\mathop{\mathrm{Tr}}\nolimits}

\newcommand{\bc}{\begin{center}}
\newcommand{\ec}{\end{center}}

\newtheorem{theorem}{Theorem}[section]
\newtheorem{lemma}[theorem]{Lemma}

\newtheorem{definition}[theorem]{Definition}

\newtheorem{corollary}[theorem]{Corollary}

\usepackage{amsfonts}

\def\01{\{0,1\}}

\newcommand{\ket}[1]{|#1\rangle}
\newcommand{\bra}[1]{\langle#1|}

\newcommand{\proj}[1]{|#1\rangle\langle#1|}

\newcommand{\N}{\mathcal{N}}

\newcommand{\comment}[1]{}
\newcommand{\diag}{\mbox{diag}}

\newcommand*{\unitarygroup}{\mathsf{U}}
 \newcommand*{\symplecticgroup}{\mathsf{Sp}}
  \newcommand*{\orthogonalgroup}{\mathsf{O}}
    \newcommand*{\passivesymplecticgroup}{\mathsf{K}}

\newcommand{\eq}[1]{\begin{align}\begin{aligned}#1\end{aligned}\end{align}}
\newcommand{\eqtext}[1]{\qquad \text{#1} \qquad}
\newcommand{\mat}[1]{\begin{pmatrix}#1\end{pmatrix}}
\DeclareMathOperator*{\trace}{\mathrm{Tr} }
\DeclareMathOperator*{\E}{\mathbb E}

\usepackage{caption}
\usepackage[titletoc,title]{appendix}

\begin{document}

\title{Typical entanglement for Gaussian states }

\author{Motohisa Fukuda \\
Faculty of Science, Yamagata University\\
1-4-12 Kojirakawa-machi, Yamagata-shi, 990-8560 Japan\\
\\
Robert Koenig \\ Institute for Advanced Study \& Zentrum Mathematik, Technical University of Munich\\
Boltzmannstr. 3, 85748 Garching, GERMANY}
\date{\today}
\maketitle
\begin{abstract}
We consider ensembles of bipartite states resulting from a random passive Gaussian unitary applied to a fiducial pure Gaussian state. We show that the symplectic spectra of the reduced density operators concentrate around that of a thermal state with the same energy. 
This implies, in particular, concentration of the entanglement entropy as well as other entropy measures. Our work extends earlier results on the typicality of entanglement beyond the two ensembles and the reduced purity measure considered in
[A. Serafini, O. Dahlsten, D. Gross, M. Plenio, J. Phys. A: Math. Theor. 40, 9551 (2007)].
\end{abstract}


\section{Introduction}
Describing quantum many-body systems is generally considered challenging because of the intrinsic complexity of general multipartite quantum states.  Nevertheless, making inferences about  ``typical'' properties of a quantum system becomes feasible if suitable notions of genericity are  considered.  A prime example is the entanglement entropy $S(\rho_A)=-\tr(\rho_A\log\rho_A)$ of a bipartite pure state $\ket{\Phi}_{AB}\in\cH_A\otimes\cH_B$. For finite-dimensional systems $\cH_A\cong\mathbb{C}^k$, $\cH_B\cong \mathbb{C}^{n-k}$, this quantity  ranges from $0$ to $\log k$ (assuming $n-k\geq k$) for general states $\ket{\Phi}$. However if $\ket{\Phi}$ is drawn according to the uniform measure on the unit sphere of $\mathbb{C}^k\otimes \mathbb{C}^{n-k}$, then~$S(\rho_A)$ is concentrated around its average. It is nearly maximal for $k\ll n$, showing that typical pure states are highly entangled.
 This kind of consideration has a long history: the expectation value of the typical entanglement was first considered by Lubkin~\cite{lubkin}, Pagels and Lloyd~\cite{pagelslloyd}, and Page~\cite{page93}. A formula for its value was rigorously proven by Foong and Kanno~\cite{foong}. In a seminal paper, Hayden, Leung and Winter~\cite{Hayden2004} systematically established the concentration around the average.
The latter work also provided the first quantum-information-theoretic applications of such typicality results, initiating a long series of paper investigating implications. For example, in~\cite{grossetal09}, it was shown that almost every state is useless for computation. 
Considerations of typical entanglement also play a key role in
our modern approach to quantum statistical mechanics, see e.g.,~\cite{popescuetal06}, 
and have been used to find phase transitions, see e.g.,~\cite{Facchi_etc}.

 The phenomenon of concentration  of entanglement is closely aligned with Jaynes' principle of maximum entropy~\cite{jaynesa,jaynesb}. The latter states that in the absence of additional information,   among all states (or distributions) consistent with given prior information, the state which maximizes the  von Neumann (respectively Shannon) entropy is distinguished. Indeed, this principle has been successfully applied in a variety of contexts involving limited prior information, ranging from hypothesis testing to statistical mechanics. Given this fact, it is natural to ask whether  Jaynes' principle can be placed on a rigorous footing in other settings where  knowledge about the system is limited. In particular, one may ask if  genericity statements can be made if only certain properties of the system such as the form of its Hamiltonian are known. Examples of this kind of study are~\cite{garneroneetal,Collins2013}, where typical entanglement in generic 1D spin chains is considered, as well as~\cite{hastingsmera}, which considers typical entanglement in random MERA-states, i.e., states restricted to belong to a certain subset of states defined by a certain variational ansatz. 
 
 Our work can be seen as an instance of this type of investigation: we are concerned with the generic features of ground states of quadratic bosonic Hamiltonians, that is, Gaussian quantum states. Such states are  ubiquitous in continuous-variable quantum information theory, and, in particular, quantum optics. The fact that Gaussian states and operations are most amenable to experimental realizations provides additional motivation for the study of typical random Gaussian states: as previously shown~\cite{serafinietalfidelity}, corresponding results have application e.g., to continuous-variable quantum teleportation.
 
 \subsubsection*{Typical entanglement in Gaussian states: prior work }
We emphasize that the study of typical entanglement in Gaussian quantum states is not new, but was pioneered by Serafini et al.~\cite{serafinidahlsten07}. In the following, we review and strengthen these results. We note that in related work~\cite{Lupo_etc}, an invariant measure on the manifold of multi-mode pure Gaussian states was studied, and the induced measure on the symplectic eigenvalues of reduced density operators was obtained. This measure is the result of acting on the vacuum state with Gaussian unitary operations. In contrast, we consider measures obtained from the orbits of general bipartite Gaussian states.   
 
Let us briefly summarize the work~\cite{serafinidahlsten07}. Analogous to the finite-dimensional case, one considers a bipartite system~$\cH_A\otimes\cH_B$ consisting of $k$~system modes  and $n-k$~environment modes (i.e., $\cH_A\cong L^2(\mathbb{R})^{\otimes k}$ and  $\cH_B\cong L^2(\mathbb{R})^{\otimes n-k}$). To formulate notions of typicality, the first issue at hand is that a normalizable Haar measure on the real symplectic group~$\symplecticgroup(2n)$ (associated with Gaussian unitaries on $\cH_A\otimes\cH_B$ via the metaplectic representation) does not exist because this group is not compact. As a consequence, the definition of a suitable ensemble of pure Gaussian state~$\ket{\Phi}\in\cH_A\otimes\cH_B$ requires introducing a compactness constraint. This usually takes the form of an upper bound
 \begin{align}
 \bra{\Phi}H_{AB}\ket{\Phi}\leq E \label{eq:energyconstraintbasic}
 \end{align}
 on the energy
 of $\ket{\Phi}$ with respect to the Hamiltonian
 \begin{align}
 H_{AB}&=\sum_{j=1}^n (Q_j^2+P_j^2)\ ,
 \end{align}
 where $(Q_1,\ldots,Q_n)$ and $(P_1,\ldots,P_n)$ are the canonical position- and momentum-operators on the system and environment, respectively (see Section~\ref{sec:gaussianpreliminaries} for details).
 
 The constraint~\eqref{eq:energyconstraintbasic} alone still does not determine a unique invariant measure over pure Gaussian states. But the covariance matrix of a Gaussian pure state satisfying~\eqref{eq:energyconstraintbasic} takes the form $OZO^T$ where 
 $O\in \symplecticgroup(2n)\cap \orthogonalgroup(2n)=:\passivesymplecticgroup(n)$ is an element of the symplectic orthogonal group $\passivesymplecticgroup(n)$ (associated with passive Gaussian unitaries), and \begin{align}
 \hat{Z}_n=
 Z_n\oplus Z_n^{-1}\ ,\qquad Z_n=\mathsf{diag}(z_1,\ldots,z_n) \label{eq:zndefrandomensemble}
 \end{align}
  is a diagonal matrix where the parameters $z_j\geq 1$ satisfy
   \begin{align}
 \frac{1}{2}\sum_{j=1}^n (z_j+z_j^{-1})\leq E\ .\label{eq:zjEcompactnessconstraint}
 \end{align}
 Indeed, the lhs.~is simply the energy $\bra{\Phi}H_{AB}\ket{\Phi}$ of $\ket{\Phi}$.
We note that our choice of parameterization differs from~\cite{serafinidahlsten07} 
by a square (i.e., they use $z_j^2$ instead of $z_j$), but this is  inconsequential.
Conversely, every pair $(O,Z)$ defines a Gaussian pure state~$\ket{\Phi}$ obeying the energy constraint~\eqref{eq:energyconstraintbasic}. Since there is a Haar measure on~$\passivesymplecticgroup(n)$, this means that any $n$-tuple of parameters $z=(z_1,\ldots,z_n)$ satisfying~\eqref{eq:zjEcompactnessconstraint} defines a probability measure~$\mu_z$ over the set of pure Gaussian states~$\ket{\Phi}$ satisfying~\eqref{eq:energyconstraintbasic}.  In fact, our work is concerned with the ensemble defined by the measure~$\mu_z$. To reiterate its definition in an operational manner: a random element of this ensemble is obtained by first preparing a pure state~$\ket{\Psi_z}$ with covariance matrix of the form~\eqref{eq:zndefrandomensemble}, and then applying a randomly chosen passive Gaussian unitary~$U_O$ corresponding to an element $O\in \passivesymplecticgroup(n)$ drawn according to the Haar measure. We note that $\ket{\Psi_z}=\ket{\varphi_{z_1}}\otimes\cdots\otimes \ket{\varphi_{z_n}}$, where each $\ket{\varphi_{z_j}}$ is a pure one-mode squeezed state with squeezing parameter~$z_j$.

  Serafini et al.~\cite{serafinidahlsten07} consider two ensembles over pure Gaussian states~$\ket{\Phi}$ satisfying~\eqref{eq:energyconstraintbasic}: they define a {\em microcanonical} measure $\mu_{\textrm{micro}}$ and {\em canonical} measure $\mu_{\textrm{canonical}}$. These can easily be expressed in terms of the distributions~$\{\mu_z\}_{z}$:
  \begin{description}
  \item[the microcanonical measure] $\mu_{\textrm{micro}}$ results by first
  drawing $(E_1,\ldots,E_n)$ uniformly (according to the measure induced by the Lebesgue measure on $\mathbb{R}^n$) 
  from the set~
  \begin{align}
  \Gamma_E=\{(E_1,\ldots,E_n)\ |\ E_j\geq 2\textrm{ for all }j=1\ldots,n\ \textrm{ and }\qquad \sum_{j=1}^n E_j \leq E\}\ ,
  \end{align}
  then setting $z_j=\frac{1}{2}(E_j+\sqrt{E_j^2-4})$ for $j=1,\ldots,n$ and drawing a pure state from~$\mu_z$. This choice of $z=(z_1,\ldots,z_n)$ amounts to distributing the available ``total'' energy $E$ uniformly over the $n$ modes in the initial product state $\ket{\Psi_z}=\ket{\varphi_{z_1}}\otimes\cdots\otimes \ket{\varphi_{z_n}}$. In particular, each state in this ensemble obeys the energy constraint~\eqref{eq:energyconstraintbasic}. 
  
   \item[the canonical measure] $\mu_{\textrm{canonical}}$
is the result of  drawing ${\bf E}=(E_1,\ldots,E_n)$ according to a 
Boltzmann distribution 
\begin{align}
dp({\bf E})&=\frac{e^{-(\sum_{j=1}^n E_j-2n)/T}}{T^n}dE_1\cdots dE_n\ 
\end{align}
with ``temperature'' $T$, where $T=E/n$, and again
setting $z_j=\frac{1}{2}(E_j+\sqrt{E_j^2-4})$ for $j=1,\ldots,n$ and drawing a pure state from~$\mu_z$. We note that states of this ensemble do not individually obey the energy constraint~\eqref{eq:energyconstraintbasic} in general, but their ensemble average does.
 \end{description} 
 
 Having defined the measures~$\mu_{\textrm{micro}}$ and  $\mu_{\textrm{canonical}}$ on the set of bipartite pure states, 
 Serafini et al.~\cite{serafinidahlsten07} compute the expected value and variance of the  inverse squared purity~$\tr(\rho_1^2)^{-2}$  of the $k=1$-mode reduced density operator $\rho_k=\tr_{n-k} \proj{\Phi}$. They analytically show that 
 the variance vanishes in the limit $n\rightarrow\infty$, implying concentration of this quantity around its average. We note that the reduced purity can be considered as a proper entanglement measure, and~\cite{serafinidahlsten07} also provide bounds relating it to the entanglement entropy $S(\rho_1)$. While these results are restricted to a single system mode, the authors of~\cite{serafinidahlsten07} also provide numerical evidence indicating that this concentration of measure also holds for $k>1$ modes.

 \subsubsection*{Our contribution} 
 In this work, we directly address concentration of measure (or more precisely, typical entanglement) for the distribution $\mu_z$ for any fixed choice of squeezing parameters $z=(z_1,\ldots,z_n)$.  This generalizes the results for the  microcanonical and the canonical measures as both of these are convolutions
\begin{align}
\mu_{\textrm{micro}}= \int  \mu_z d\nu_{\textrm{micro}}(z)\qquad\qquad 
\mu_{\textrm{canonical}}=\int  \mu_z d\nu_{\textrm{canonical}}(z)
\end{align}
 of $\mu$ with certain distributions $\nu_{\textrm{micro}}$ and $\nu_{\textrm{canonical}}$ on the set $\{(z_1,\ldots,z_n)\ |\ z_j\geq 1\textrm{ for all }j=1,\ldots,n\}$. Our analysis  also generalizes existing results in the following ways:
 \begin{enumerate}[(i)]
 \item
 We consider $k$-mode reduced density operators of pure Gaussian states of $n$~modes. Instead of considering a fixed number of modes (such as $k=1$), we allow $k$ to grow with $n$ at a polynomial rate (see Definition~\ref{definition:bounded-subsystem}).
 
 \item
The only assumption on the sequence~$z=(z_1,\ldots,z_n)$ of squeezing parameters we need is an upper bound on the maximal value $\max_{1\leq j\leq n} z_j$. We note that this corresponds to the amount of squeezing and is thus equivalent to the maximal energy concentrated in a single mode in the initial state~$\ket{\Psi_z}$.  In fact, we can allow this quantity to grow at a polynomial rate in the number of modes~$n$. In particular, with all but an exponentially small probability, a sequence $z$ drawn from the distribution $\mu_{\textrm{canonical}}$ will satisfy this constraint.

 \item
 Our analysis focuses on the symplectic eigenvalues of the covariance matrix defined by the $k$-mode reduced density operator~$\rho_k$.
 We show that these concentrate around the ``flat'' symplectic spectrum of a thermal state with the same energy. As the symplectic eigenvalues directly determine the entanglement spectrum (spectrum of the local state) of the state~$\ket{\Phi}$,  this immediately implies concentration of all (suitably continuous) entanglement measures. We exemplify this using the entanglement entropy.
 \end{enumerate}
Our approach shares some similarities, but also some key differences to that of~\cite{serafinidahlsten07}: Instead of considering all symplectic invariants, we compute certain  second and fourth moments of the (random) covariance matrix of the reduced state. Since this involves only integration over the Haar measure~$\passivesymplecticgroup(n)$ (which happens to be isomorphic to the unitary group~$\unitarygroup(n)$), we can derive explicit expressions for these moments using the Weingarten calculus. We then show that these expressions can be used to obtain our concentration results.

 \subsubsection*{Outline}
 In Section~\ref{sec:gaussianpreliminaries}, we review some basic notions related to Gaussian quantum states and operations.  Our main results are shown in Section~\ref{sec:mainresults}. In Section~\ref{sec:measuresbipartitepurestates}, we specify our probability measure and the main assumptions. In Section~\ref{sec:moments}, we compute the relevant moments.  Our main conclusions are derived subsequently: Theorem~\ref{theorem:thermal} provides explicit tail bounds on the deviation of the symplectic spectrum from that of a thermal state. Theorem~\ref{sec:thermalizationintermsofentropy} gives the corresponding statement for entanglement entropy.

\section{Basics on Gaussian states and operations\label{sec:gaussianpreliminaries}} 
We briefly review the pertinent facts about bosonic quantum systems, and refer to the literature for more details. Throughout, we consider a continuous-variable quantum system with $n$~canonical degrees of freedom (or {\em modes}). Such a system is described by the Hilbert space $\cH_n=L^2(\mathbb{R})^{\otimes n}$.
Let $R=(Q_1,\ldots,Q_n,P_1,\ldots,P_n)$ be the canonical position- and momentum operators. These satisfy the canonical commutation relations
\begin{align}
[R_j,R_k]=i J_{j,k} I_{\cH_n}\qquad\textrm{ for }\qquad j,k\in \{1,\ldots,2n\}\ ,
\end{align}
where $J=\begin{pmatrix}
0_n  & -I_n\\
I_n & 0_n
\end{pmatrix}$ is the matrix defining the symplectic form (here $I_n$ and $0_n$ are the $n\times n$ identity and zero matrix, respectively). 

\subsection{Gaussian states}
The set of {\em Gaussian states} is characterized by
having Gaussian characteristic functions; for our purposes, the main property we need is that a Gaussian state $\rho$ on $\cH_n$ is fully determined by its first and second moments
\begin{align}
d_r &=\tr(\rho R_r)\qquad\textrm{for }\qquad r=1,\ldots,2n\\
M_{r,s}&=\tr(\rho \{R_r-d_r \mathsf{id}_{\cH_n},R_s-d_s \mathsf{id}_{\cH_n}\})\qquad\textrm{for }\qquad r,s\in \{1,\ldots,2n\}\ .
\end{align}
Here $\{A,B\}=AB+BA$ denotes the anticommutator. The vector $d=(d_1,\ldots,d_{2n})$ is sometimes called the {\em displacement}, whereas $M=(M_{r,s})_{r,s=1}^{2n}$ is the {\em covariance matrix} of $\rho$. Since states with different displacement vectors (but identical covariance matrix) are related by tensor products of local (single)-mode unitary operations (so-called displacement or Weyl operators), the displacement plays no role in entanglement considerations. Therefore, we will restrict our attention to {\em centered} Gaussian states, i.e., states whose displacement vector vanishes.

Valid covariance matrices (i.e., those associated with a Gaussian quantum state) are those symmetric matrices~$M$ satisfying the uncertainty relation (operator inequality) $M+iJ\geq 0$. Any such  covariance matrix $M$ can be brought into {\em Williamson normal form}: there is an element $S\in \symplecticgroup(2n)$ of the real symplectic group $\symplecticgroup(2n)=\{S\in \mathsf{Mat}_{2n\times 2n}(\mathbb{R})\ |\ SJS^T=J\}$ such that 
\begin{align}
SMS^T=\mathsf{diag}(\lambda_1,\lambda_1,\lambda_2,\lambda_2,\ldots,\lambda_n,\lambda_n)\ .\label{eq:symplecticdiagonalization}
\end{align}
The parameters $\lambda_j$ satisfy $\lambda_j\geq 1$ and are called the {\em symplectic eigenvalues } of $\rho$. They can be obtained by computing the spectrum of the matrix~$JM$,  which consists of complex conjugate pairs as follows:
\begin{align}
\label{eq:random_ev}
\mathsf{spec}(JM)= \bigcup_{j=1}^n \{\pm i \lambda_j\}\ .
\end{align}

As discussed below, symplectic group elements correspond to unitary Gaussian operations on the Hilbert space. Therefore, Eq.~\eqref{eq:symplecticdiagonalization} implies that the symplectic eigenvalues fully determine the spectrum of the Gaussian state.  More precisely,
in a suitably chosen basis, the Gaussian state is a product state of single-mode thermal states with mean photon number
\begin{align}
\N(\lambda)=(\lambda-1)/2\ 
\end{align}
for each symplectic eigenvalue~$\lambda$. In particular, we may express unitarily invariant functions such as the von Neumann entropy $S(\rho)=-\tr(\rho\log\rho)$ of an $n$-mode Gaussian state $\rho$ as  functions of the symplectic eigenvalues: We have 
\begin{align}\label{eq:entropy}
S(\rho)=\sum_{j=1}^n G(\lambda_j)\qquad \qquad\textrm{ where }\qquad 
\begin{matrix}
G(\lambda)&=&g(\N(\lambda_j))\qquad \textrm{ and }\\
g(\N)&=&(\N+1)\log(\N+1)-\N\log \N\ . 
\end{matrix}
\end{align}
Slightly abusing notation, we will write
$S(M)$ for the von Neumann entropy $S(\rho_M)=-\tr(\rho_M\log \rho_M)$ of a Gaussian state $\rho_M$ with covariance matrix~$M$. 
 For example, the Gaussian state with covariance matrix $M=\lambda I_{2n}$ (with $\lambda\geq 1$) has entropy
\begin{align*}
S(\lambda I_{2n})&=n\cdot G(\lambda)\ .
\end{align*}
Such states (with covariance matrix proportional to the identity) are called thermal states; they are Gibbs states of the Hamiltonian 
\begin{align}
H_0 &=\sum_{j=1}^n (Q_j^2+P_j^2)\ .\label{eq:hzeroHamiltonian}
\end{align}
Note that for a general state~$\rho$ with covariance matrix $M$, we have $\tr(H_0\rho)=\frac{1}{2}\tr(M)$.

The covariance matrix $M$ of a pure $n$-mode Gaussian state $\rho=\proj{\Psi}$ can be diagonalized by using an orthogonal symplectic matrix $O\in \passivesymplecticgroup(n):=\symplecticgroup(2n)\cap \orthogonalgroup(2n)$ (where $\orthogonalgroup(2n)=\{O\in\mathsf{Mat}_{2n\times 2n}(\mathbb{R})\ |\ O^TO=OO^T=I_{2n}\}$), namely
\begin{align}
\label{eq:pure-formula}
M = O \begin{pmatrix}Z_n &0 \\ 0 & Z_n^{-1} \end{pmatrix} O^T\qquad\textrm{ where }\qquad 
Z_n = \mathrm{diag} (z_1, \ldots , z_n)\ \textrm{ with }z_j\geq 1\textrm{ for all }j=1,\ldots n\ .
\end{align}
We call the entries $\{z_j\}_{j=1}^n$ of the diagonal matrix $Z_n$ the {\em squeezing parameters}. 
Eq.~\eqref{eq:pure-formula} follows from the Euler/Bloch-Messiah decomposition of a general symplectic transformation
(this is nothing but the singular value decomposition of a symplectic transformation). Observe that such a state has energy
\begin{align}
\bra{\Psi}H_0\ket{\Psi}&=\sum_{j=1}^n E_j\qquad\textrm{where }\qquad E_j=\frac{1}{2}(z_j+1/z_j)\qquad\textrm{ for }j=1,\ldots,n\ .
\end{align}
We can understand $E_j$ as the energy in the $j$-th mode after a suitable Gaussian unitary rotation.

\subsection{Gaussian operations: The partial trace}
Gaussian operations are completely positive trace-preserving operations on $\cB(\cH_n)$ which preserve the set of Gaussian states. An example is  the partial trace 
\begin{align}
\begin{matrix}
\tr_{n-k}: &\cB(\cH_n)&\rightarrow & \cB(\cH_k)\\
 & \rho_n & \mapsto  & \rho_k:=\tr_{n-k}\rho_n
\end{matrix}
\end{align}
which maps an $n$-mode quantum state~$\rho_n$ to its reduced $k$-mode density operator $\rho_k$ on the first $k<n$ modes. The effect of this operation on Gaussian states can be described by its action on covariance matrices: if $\rho_n$ has the covariance matrix $M_n$, then the reduced density operator $\rho_k$ has covariance matrix~$M'_{n,k}=(M'_{n,k})_{r,s=1}^{2k}$ defined by the entries 
\begin{align}
(M'_{k})_{r,s}&=
\begin{cases}
(M_{n,k})_{r,s}\qquad&\textrm{ if }  r\leq k, s\leq k\\
(M_{n,k})_{r,n+(s-k)}\qquad&\textrm{ if }  r\leq k, s>k\\
(M_{n,k})_{n+(r-k),s}\qquad&\textrm{ if }  r>k, s\leq k\\
(M_{n,k})_{n+(r-k),n+(s-k)}\qquad&\textrm{ if }  r>k, s> k
\end{cases} 
\end{align}
for $r,s\in \{1,\ldots,2k\}$. In other words, the covariance matrix~$M'_{k}$ of the reduced density operator simply corresponds to taking the submatrix associated with the observables $\{Q_r\}_{r=1}^k\cup \{P_r\}_{r=1}^k$, which amounts to taking a submatrix of $M_{n}$. In the following, we will -- for simplicity of notation -- identify~$M_{n,k}$ with  the $2n\times 2n$-matrix obtained by 
keeping these entries of $M_{n}$ but setting all other entries to~$0$.
We can then write
\begin{align}
M'_{n,k} = \hat \Pi_{n,k} M_n \hat \Pi_{n,k}  \label{eq:partialtracecovariancematrix}
\end{align}
for the projection
\eq{
\hat \Pi_{n,k} = \mat{\Pi_{n,k}&0_n\\ 0_n& \Pi_{n,k}}\ ,
}
where $\Pi_{n,k}=\diag(\underbrace{1,\ldots,1}_{k\textrm{ times }},\underbrace{0,\ldots,0}_{n-k\textrm{ times}})$ is an $n \times n$ projection matrix of rank $k$. 

\subsection{Passive Gaussian unitary operations and the Haar measure}
 The set of Gaussian unitaries on $\cH_n$ (i.e., unitary maps  $U:\cH_n\rightarrow\cH_n$ defining  Gaussian operations $\Phi_U(\rho)=U\rho U^*$)  is in one-to-one correspondence with the real symplectic group: any $S\in \symplecticgroup(2n)$ defines a Gaussian unitary $U_S$ on $\cH_n$. The action of such a unitary~$U_S$ on a (centered) Gaussian state~$\rho$ is described by the map
\begin{align}
M\mapsto M'=SMS^T\ ,\label{eq:transformationcovariancematrices}
\end{align}
which maps the covariance matrix $M$ of $\rho$ to the covariance matrix $M'$ of the rotated state $\rho'=U_S\rho U_S^*$.

The group $\symplecticgroup(2n)$ is not compact, but the group $\passivesymplecticgroup(n):=\symplecticgroup(2n)\cap \orthogonalgroup(2n)$ of orthogonal symplectic matrices is. Gaussian unitaries $U_O$ on $\cH_n$ associated with elements $O\in \passivesymplecticgroup(n)$ are called {\em passive}. They have the property that they preserve the Hamiltonian~\eqref{eq:hzeroHamiltonian}, that is, $U_O^* H_0 U_O=H_0$. Passive Gaussian unitaries can be implemented using  phase shifters and beamsplitters only as shown in~\cite{recketal}. 

Contrary to the set of general Gaussian unitaries, it is possible to define an invariant measure on the set of passive Gaussian unitaries. The construction uses the fact that the group $\passivesymplecticgroup(n)$ is isomorphic to the (complex) unitary group $\unitarygroup(n)=\{U\in \mathbb{C}^{n\times n}\ |\ U^* U=I_n\}$, with isomorphism given by
\begin{align}
\eta: \unitarygroup(n) & \rightarrow &&\passivesymplecticgroup(n)=\symplecticgroup(2n)\cap \orthogonalgroup(2n)\\
U & \mapsto &&\begin{pmatrix} \mathsf{Re}(U) & \mathsf{Im}(U)\\ -\mathsf{Im}(U) & \mathsf{Re}(U) \end{pmatrix} 
\equiv 
F  \begin{pmatrix}  U & 0\\ 0 & \overline{U} \end{pmatrix}F^{-1}\qquad\textrm{ where }
\qquad F=\frac{1}{\sqrt{2}} \begin{pmatrix} I_{n} & iI_{n} \\ iI_{n} & I_{n} \end{pmatrix}.
\label{eq:isomorphismtosymplecticgroup}
\end{align}
We refer to e.g.,~\cite{degosson} for more information about this isomorphism and the symplectic group in general. The isomorphism~\eqref{eq:isomorphismtosymplecticgroup} immediately gives rise to the notion of a Haar measure of $\passivesymplecticgroup(n)$: indeed, the Haar measure on the  unitary group $\unitarygroup(n)$ induces a measure over $\passivesymplecticgroup(n)$ (which is obviously left- and right-invariant). This in turn defines a measure on the set of passive Gaussian unitaries (via $O\mapsto U_O$), and it is this latter measure we use to define ensembles of pure Gaussian states.

\section{Concentration of measure for pure Gaussian states\label{sec:mainresults}}
Here we establish our main results. In Section~\ref{sec:measuresbipartitepurestates}, we formally introduce the distributions we consider, and provide the main definitions. In Section~\ref{sec:moments}, we discuss
how to compute certain moments of the (random) covariance matrix of the reduced density operators. 
In Section~\ref{sec:typicalitysymplectic}, we derive our main concentration result for the symplectic spectra of these operators. Finally, in Section~\ref{sec:entropyreduceddensity}, we show how these imply concentration results for the von Neumann entropy.
\subsection{Ensembles of bipartite pure Gaussian states and their description\label{sec:measuresbipartitepurestates}}
We now formally define the measure over pure Gaussian states on~$\cH_n$ we consider, and express it in terms of covariance matrices. Fix a sequence  $(z_1,\ldots,z_n)\in [1,\infty)^n$ of the squeezing parameters and let $Z_n=\mathsf{diag}(z_1,\ldots,z_n)$. Let $\ket{\Psi_{Z_n}}\in\cH_n$ be the $n$-mode Gaussian state with  covariance matrix~$Z_n\oplus Z_n^{-1}\in\mathsf{Mat}_{2n\times 2n}(\mathbb{R})$, i.e., $\ket{\Psi_{Z_n}}$ is a tensor product of single-mode squeezed states.
 The distribution~$\mu_{Z_n}$ over pure Gaussian states then is obtained by
\begin{enumerate}[(i)]
\item
drawing an element $U\in \unitarygroup(n)$ at random (from the Haar measure) and
\item
outputting the state 
\begin{align}
\ket{\Psi(U)}=U_{\eta(U)}\ket{\Psi_{Z_n}}\label{eq:PsiUdef}
\end{align} obtained by applying
the passive Gaussian unitary~$U_{\eta(U)}$ associated with the element $\eta(U)\in \passivesymplecticgroup(n)$.
\end{enumerate}
In other words, the Haar measure on $\unitarygroup(n)$ induces a measure~$\mu_{Z_n}$ over pure Gaussian states via Eq.~\eqref{eq:PsiUdef}. We will be interested in the partial traces $\tr_{n-k}\proj{\Psi(U)}$ of these states. 

Expressed in terms of covariance matrices, the state
$\ket{\Psi(U)}$ has covariance matrix $M_{n,k}=M_{n,k}(U)$ given by (cf.~\eqref{eq:transformationcovariancematrices}).
\begin{align}\label{eq:rotated-M}
M_{n}(U)=\eta(U) (Z_n\oplus Z_n^{-1})\eta(U)^T\ .
\end{align}
In particular, the reduced density operator~$\tr_{n-k}\proj{\Psi(U)}$ has covariance matrix (cf.~\eqref{eq:partialtracecovariancematrix})
\begin{align}
M_{n,k}(U)=\hat{\Pi}_{n,k}M_n(U)\hat{\Pi}_{n,k}\ .\label{eq:mnkudef}
\end{align}
We are interested in the symplectic eigenvalues of $M_{n,k}(U)$ (for a typical choice of $U$), hence we will consider (cf.~\eqref{eq:random_ev}) the random matrix
\begin{align}
JM_{n,k}(U)&=\hat{\Pi}_{n,k}J M_n(U) \hat{\Pi}_{n,k}=\hat{\Pi}_{n,k}J \eta(U) (Z_n\oplus Z_n^{-1})\eta(U)^T \hat{\Pi}_{n,k}\ . \label{eq:jmnkudef}
\end{align}
where $U \in \unitarygroup(n)$ is chosen with respect to Haar measure on the unitary group. (Here we used the particular block-diagonal form of $\hat{\Pi}_{n,k}$ in the first identity.)

\subsubsection*{Sequences of distributions and reduced density operators}
In the following, we will be interested in asymptotic properties 
of the reduced density operators as a function of the number of modes~$n$. Correspondingly, we  consider sequences of $n$-tuples of squeezing parameters  (for different $n$) associated with $n$-mode pure product squeezed states $\ket{\Psi_{Z_n}}\in\cH_n$. The following definition will be convenient:
\begin{definition}\label{definition:bounded-squeeze}
Let $\{Z_n\}_{n\in\mathbb{N}}$ be a sequence of matrices of the form
\begin{align}
Z_n=\mathsf{diag}(z_1^{(n)},\ldots,z_{n}^{(n)})\qquad\textrm{ with }\qquad z_j^{(n)}\geq 1\qquad\textrm{ for all }1\leq j\leq n\textrm{ and }n\in\mathbb{N}\ .
\end{align}
We say that the sequence $\{Z_n\}_{n\in\mathbb{N}}$ has bounded squeezing of degree $\zeta\geq 0$ if there is a constant $C>0$ such that 
\begin{align}
\|Z_n \|_\infty \leq C n^\zeta\qquad\textrm{ for all large enough }n\in\mathbb{N}\ .\label{eq:bounded-squeeze}
\end{align}
\end{definition}
With $\zeta=0$, this definition includes the physically most relevant case where each single mode has a bounded amount of energy (which is  equivalent to squeezing for one-mode squeezed states). However, it also permits considerations of scenarios where the maximal amount of squeezing per mode can grow with the number of modes.

A sequence $\{Z_n\}_{n\in\mathbb{N}}$ of such squeezing parameters 
defines a sequence of distributions $\{\mu_{Z_n}\}_{n\in\mathbb{N}}$ over  $n$-mode pure Gaussian states, and we can consider their $k$-mode reduced density matrices. Again, we may let $k=k_n$ depend on $n$. Let us call the number of modes $k_n\in \{1,\ldots,n\}$ the ``subsystem size''. We use the following definition:
\begin{definition}\label{definition:bounded-subsystem}
A sequence $\{k_n\}_{n\in\mathbb{N}}$ of subsystem sizes is called {\em bounded of degree~$\kappa$} if there is a constant $K>0$ such that 
\begin{align}
k_n\leq  Kn^\kappa\qquad\textrm{ for all  large enough } n\in \mathbb{N}\ .
\end{align}
\end{definition}
Similarly as for squeezing, this includes the case (with $\kappa=0$) where 
the subsystem size is held constant. However, our derivation also permits some degree of growth with the total number of modes~$n$.

\subsection{Moments of typical Gaussian pure states\label{sec:moments}}
Here we analyze typical random Gaussian pure states. Our main result concerns typicality of 
the~$k_n$-mode reduced density operators of pure bipartite states of $n$~modes, distributed according to the distribution $\mu_{Z_n}$. More precisely, we investigate the asymptotic behavior of the symplectic eigenvalues~$\{\lambda^{(n)}_j\}_{j=1}^{k_n}$  of the  covariance matrices $M_{n,k}$ (cf.~Eq.~\eqref{eq:mnkudef}). We assume here  that $\{Z_n\}_{n\in\mathbb{N}}$ has bounded squeezing of degree~$\zeta\geq 0$ (see Def.~\ref{definition:bounded-squeeze}), and the sequence~$\{k_n\}_{n\in\mathbb{N}}$ of subsystem sizes is bounded of degree~$\kappa$ (as in Def.~\ref{definition:bounded-subsystem}), for suitably chosen $\zeta\geq 0$ and $\kappa\geq 0$.  We will show  (see Theorem~\ref{theorem:thermal}) that all the symplectic eigenvalues $\{\lambda^{(n)}_j\}_{j=1}^{k_n}$ are typically close to the average energy~$\bra{\Psi_{Z_n}}H_0\ket{\Psi_{Z_n}}$ of the state~$\ket{\Psi_{Z_n}}$ (see Section~\ref{sec:measuresbipartitepurestates}), that is,   to 
\eq{\label{eq:agerage-energy}
\lambda(n)= \frac{\tr \hat Z_n}{2n}=\frac{1}{2n}\sum_{j=1}^n (z^{(n)}_j+1/z^{(n)}_{j})\ .
}
as $n \to \infty$, assuming that $\lambda(n)$ is uniformly bounded in $n$.

In preparation for this result, we compute certain moments of the (random) covariance matrix~$M_{n,k}$. We begin with the average of~$( JM_{n,k} )^2$:
\begin{theorem}\label{theorem:average}
Let $\{Z_n\}_{n\in\mathbb{N}}$ be a sequence of matrices with bounded squeezing of degree~$0\leq \zeta<1$ as  in Definition \ref{definition:bounded-squeeze}. Assume that  average energy \eqref{eq:agerage-energy} is uniformly bounded with respect to $n$. Let $\{k_n\}_{n\in\mathbb{N}}$ be a sequence of subsystem sizes bounded of degree $0 \leq \kappa <1$ as in Definition \ref{definition:bounded-subsystem}.
Then we have 
\begin{align}
\E (JM_{n,k})^2 = \left( - \lambda(n)^2 +O(n^{\kappa+\zeta-1}) \right) I_{2k} \label{eq:jmnksquare}
\end{align}
where $M_{n,k}=M_{n,k}(U)$ (cf.~\eqref{eq:mnkudef}) and where the expectation is taken over $U$ chosen uniformly from the Haar measure on $\unitarygroup(n)$. 
\end{theorem}

\begin{proof}
Recall from~\eqref{eq:jmnkudef} that 
\eq{
JM_{n,k} &=\hat \Pi_{n,k} J \eta(U) \hat{Z}_n \eta(U)^T \hat \Pi_{n,k} 
=\frac 12  \mat{i\Pi &\Pi \\-\Pi &-i\Pi }\mat{UAU^T&-iUBU^*\\-i\bar U BU^T & -\bar UAU^*}\mat{\Pi &i\Pi \\i\Pi &\Pi }
}
where
\eq{
A = \frac {Z_n-Z_n^{-1}}{2} \eqtext{and}  B= \frac {Z_n+Z_n^{-1}}{2}\ .
}
Here we inserted the particular form~\eqref{eq:isomorphismtosymplecticgroup} of the isomorphism~$\eta$. 
Squaring this matrix then gives
\eq{
(JM_{n,k})^2 = \frac 12 \mat{i\Pi &\Pi \\-\Pi &-i\Pi } \Theta    \mat{\Pi &i\Pi \\i\Pi &\Pi }
}
Here, $\Theta=\Theta(U) $ is a random $2n\times 2n$ matrix, which is the random part of $(JM_{n,k})^2$. It is given by
\eq{
\Theta  = \mat{iUBU^*\Pi UAU^T-iUAU^T\Pi \bar UBU^T&UBU^*\Pi UBU^*-UAU^T\Pi \bar UAU^*\\
\bar UAU^*\Pi UAU^T-\bar UBU^T\Pi \bar UBU^T&-i\bar UAU^*\Pi UBU^*+i\bar UBU^T\Pi \bar UAU^*}\ .
}
As the entries are polynomials in $U, U^*, U^T, \bar{U}$, they can easily be computed using the Weingarten calculus (see Section~\ref{sec:weingarten}). In particular, it is easy to see that when one takes the average over the unitary group, the diagonal blocks of  $\Theta $ vanish such that 
\eq{
\E [\Theta]  = \mat{0& X\\-\bar X&0} = \mat{0& X\\- X&0}
}
where
\eq{
X = \E \left[UBU^*\Pi UBU^*-UAU^T\Pi \bar UAU^*\right]\ , 
}
which is a real $n\times n$ matrix.
Therefore, reinserting this into~\eqref{eq:jmnksquare} gives 
\eq{
\E \left[(JM_{n,k})^2\right]= - \mat{\Pi X\Pi &0\\0&\Pi X\Pi }
}
Computing the matrix $X$ using the Weingarten calculus  (see Lemma \ref{lemma:calculation-average}) we get
\eq{
\E [(JM_{n,k})^2 ] = -\tilde \lambda^2 \mat{\Pi &0\\0&\Pi } = -\tilde\lambda^2 \hat \Pi 
}
where 
\begin{align}
\tilde \lambda^2 &= \frac{1}{n^2-1} \left[ (\tr B)^2 - \tr A^2 + \tr \Pi  \cdot (\tr B^2 -\tr A^2)\right]\\
& \hspace{2cm}+ \frac{1}{n(n^2-1)} \left[\tr A^2 -\tr B^2 +  \tr \Pi  \cdot (\tr A^2 - (\tr B)^2) \right] \\
&= \left(\frac{n-k}{n(n^2-1)}\right)(\tr B)^2-\left(\frac{k+1}{n(n+1)}\right)\tr(A^2)+\left(\frac{kn-1}{n(n^2-1)}\right)\tr(B^2)\ .\label{eq:secondmomentpolynomial}
\end{align}
Let us analyze asymptotic behavior of the quantity~$\tilde{\lambda}^2$. 
Using the assumption that
$k=k_n=O(n^\kappa)$, we have
\begin{align}
\tilde \lambda^2 &=\frac{1}{n^2}(1+O(n^{\kappa-1})) (\tr B)^2+O(n^{\kappa-2})\tr(A^2)+O(n^{\kappa-2})\tr(B^2)
\end{align}
Note that we have 
\eq{
 \trace [A^2] \leq  \trace [B^2] \leq \| B \|_1 \|B\|_\infty \leq C\lambda(n)\cdot  n^{1+ \zeta}
}
because 
\begin{align}
\|B\|_1 = \trace B = \lambda(n)\cdot n\label{eq:boneinsert}
\end{align} and $\|B\|_\infty \leq C n^\zeta$ for some constant $C>0$ by assumption.  
Since $\lambda(n)$ is uniformly bounded in $n$,  we obtain
\eq{
\tilde \lambda ^2 =
\lambda(n)^2+O(n^{\kappa+\zeta-1})
}
by straightforward computation.  This is the claim.
\end{proof}
In the following, we will only need the following immediate corollary:
\begin{corollary}[Second moment]\label{cor:secondmoments}
Let $\{Z_n\}_{n\in\mathbb{N}}$ and $\{k_n\}_{n\in\mathbb{N}}$ be as in Theorem~\ref{theorem:average}. Then 
\begin{align}
\E\left[\tr((JM_{n,k})^2)\right]&=-2k_n \lambda(n)^2+O(n^{2\kappa+\zeta-1})\ .
\end{align}
\end{corollary}

We will also need an analogous statement for the fourth moment. This can be derived in a similar manner although with more effort: Explicit expressions can be obtained from the graphical Weingarten calculus.
The asymptotics of these expressions may then be bounded using the following estimate: we have
\begin{align}
\trace \left[\prod_{i=1}^m  X_i \right] 
\leq \left\| X_1 \right\|_1  \left\| \prod_{i=2}^m  X_i  \right\|_\infty
\leq \left\| B \right\|_1  \left\| B \right\|_\infty^{m-1} 
\leq \lambda(n) C n^{1+ (m-1)\zeta}\label{eq:asymptoticboundv}
\end{align}
where $X_i\in \{A,B\}$ and where we inserted the expression~\eqref{eq:boneinsert}. We defer this computation to Appendix~\ref{appendix:moments}. Here we only state the result we will need:

\begin{lemma}\label{lem:fourthmomentsexplicit}
Let $\{Z_n\}_{n\in\mathbb{N}}$ and $\{k_n\}_{n\in\mathbb{N}}$ be as in Theorem~\ref{theorem:average}. Then 
\begin{align}
\E [\tr\left((JM_{n,k})^4\right) ]=2k_n \lambda(n)^4 + O(n^{2\kappa+\zeta-1})
\end{align}
\end{lemma}
We will henceforth restrict our attention to the regime $2\kappa+\zeta \leq 1$: here both Theorem~\ref{theorem:average} 
(respectively Corollary~\ref{cor:secondmoments}) and Lemma~\ref{lem:fourthmomentsexplicit} can be applied.

\subsection{Typical behavior of symplectic eigenvalues in reduced density operators\label{sec:typicalitysymplectic}}
We show concentration of symplectic eigenvalues in the setting of Theorem \ref{theorem:average}, i.e.,  we assume that we have 
\begin{enumerate}[(i)]
\item a sequence of diagonal matrices $\{Z_n\}_{n\in\mathbb{N}}$ with bounded squeezing of degree $0 \leq \zeta<1$, as in Definition,  \ref{definition:bounded-squeeze}
\item an average energy~$\lambda(n)$ (cf.~\eqref{eq:agerage-energy}) uniformly bounded as $n \to \infty$, and 
\item a sequence~$\{k_n\}_{n\in\mathbb{N}}$  of subsystem sizes which is bounded of degree $0 \leq \kappa <1$, as in Definition~\ref{definition:bounded-subsystem}.
\end{enumerate}
To study the symplectic spectrum of the matrices $M_{n,k}(U)$, we are interested in the function
\eq{\label{equation:function}
\begin{matrix}
f: &\unitarygroup(n)&\rightarrow &\mathbb{R}\\
&U&\mapsto &f(U):= \trace\left[ \left((JM_{n,k}(U))^2 + \lambda(n)^2 I_{2k}  \right)^2 \right]\ .
\end{matrix}
}
The point here is that $f(U)$ directly quantifies the difference between the symplectic spectrum of $M_{n,k}(U)$ and
a ``flat'' spectrum with all symplectic eigenvalues equal to the average energy~$\lambda(n)$. This is expressed by the following lemma: 
\begin{lemma}\label{lemma:f-formula}
Let $f:\unitarygroup(n)\rightarrow\mathbb{R}$ be defined by~\eqref{equation:function}, let $\lambda(n)$ denote the average energy~\eqref{eq:agerage-energy} and let $\{\lambda_j(U)\}_{j=1}^k$ denote the symplectic eigenvalues of $M_{n,k}(U)$. Then 
\eq{
f(U) = 2 \sum_{j=1}^k \left(\lambda_j(U)^2 - \lambda(n)^2\right)^2 \ .
}
\end{lemma}
\begin{proof}
Omitting the $U$-dependence for ease of notation, we first note that  because of  property \eqref{eq:random_ev}, 
the matrix $JM_{n,k}$ can be diagonalized using with some invertible matrix $V$, that is, 
\eq{
JM_{n,k} = V \Lambda V^{-1}\ ,
}
where $\Lambda = \mathrm{diag} \{ \pm i \lambda_j\}$. With the cyclicity of the trace it follows that
\eq{
f(U) = \trace \left[ \left( V \Lambda^2 V^{-1} + \lambda(n)^2 I_{2k}   \right)^2 \right] 
=  \trace \left[ \left(\Lambda^2 + \lambda(n)^2 I_{2k}   \right)^2 \right]
}
The claim follows from this.
\end{proof}
Our choice to consider the quantity $f(U)$ (cf.~\eqref{equation:function}) 
 is motivated by Lemma~\ref{lemma:f-formula}, which connects it to the deviation of the squares~$\lambda_j(U)^2$ of the symplectic eigenvalues  from the square $\lambda(n)^2$  of the average energy.
Let us make a remark that other quantities such as~$\| (JM_{n,k}(U))^2 + \lambda(n)^2 I_{2k}\|_2$ would not immediately provide such expressions because the matrix in the $2$-norm is not Hermitian. We next bound the average of this quantity:

\begin{lemma}\label{lemma:average-f}
Let  $f:\unitarygroup(n)\rightarrow\mathbb{R}$ be  the function given by~\eqref{equation:function}. Then there is a universal constant $C>0$ such that 
\eq{
\mathbb E f(U)  \leq C n^{2 \kappa+\zeta-1 }\ ,
}
where the average is taken with respect to the Haar measure on $\unitarygroup(n)$. 
\end{lemma}
\begin{proof}
We compute
\begin{align}
\E \trace \left[ \left( (JM)^2 + \lambda^2 I  \right)^2 \right] 
= \E \trace  \left[ (JM)^4 \right]+ 2 \lambda^2 \E \trace  \left[(JM)^2\right]  + 2 k \lambda^4
\label{eq:ejmlamdasquared}
\end{align}
where $\lambda = \lambda(n)$ is from \eqref{eq:agerage-energy} and $I=I_{2k}$.
Invoking Corollary~\ref{cor:secondmoments} and Lemma~\ref{lem:fourthmomentsexplicit} we obtain that this is of order $O(n^{2k+\zeta-1})$. 
\end{proof}

In the following, we will need an estimate on the continuity of the function~$f$. For this purpose, we use the
norm $\|A\|_2=\sqrt{\tr(A^*A)}$ to measure distance on the unitary group. 
\begin{lemma}\label{lemma:Lipschitz}
The  Lipschitz constant of the function $f:\unitarygroup(n)\rightarrow\mathbb{R}$, defined by~\eqref{equation:function}, is upper bounded by $C n^{4 \zeta + \frac \kappa 2}$, where $C>0$ is a universal constant. 
\end{lemma}
\begin{proof}
Let $M$ and $L$ be two covariance matrices obtained from two unitaries $U,V\in \unitarygroup(n)$, respectively, that is, $M=M_{n,k}(U)$ and $L=M_{n,k}(V)$.
Then
\begin{align}
|f(U) - f(V)| &\leq \left|\trace \left[(JM)^4 - (JL)^4\right]\right| + 2 \lambda^2  \left |\trace \left[(JM)^2 - (JL)^2\right]\right|\ .\label{eq:fufvdiff}
\end{align}
We can bound the first term using the triangle inequality as follows: We have 
\begin{align}
\left|\trace \left[(JM)^4 - (JL)^4\right]\right|& \leq  \|(JM)^4-(JL)^4\|_1\\
&\leq \|(JM)^4-(JM)^3(JL)\|_1+\|(JM)^3(JL)-(JM)^2(JL)^2\|_1\\
&\qquad +\|(JM)^2(JL)^2-(JM)(JL)^3\|_1+\|(JM)(JL)^3-(JL)^4\|_1\\
&\leq \|M\|_\infty^3 \|M-L\|_1+\|M\|_\infty^2 \|M-L\|_1 \|L\|_\infty\\
&\qquad +\|M\|_\infty \|M-L\|_1\|L\|_\infty^2+\|M-L\|_1\|L\|_\infty^3\ .
\end{align}
In the last inequality, we used the identity $\|J\|_\infty=1$ and the inequality 
\eq{\label{eq:matrix-inequality}
\|ABC\|_p\leq \|A\|_\infty \|B\|_p \|C\|_\infty.
}
for $p \in [1,\infty]$, which can be proved as follows. For $p \in [1,\infty)$
\eq{
\|ABC\|_p^p &= \trace \left[ (ABCC^*B^*A^* )^{\frac p2}\right] \leq \|C\|_\infty^p \trace \left[ (ABB^*A^* )^{\frac p2}\right] \\
&= \|C\|_\infty^p \trace \left[ (B^*A^*AB )^{\frac p2}\right]  \leq  \|C\|_\infty^p  \|A\|_\infty^p   \|B\|_p^p 
}
Moreover, with \eqref{eq:matrix-inequality}, the fact that $\eta(U)$ is orthogonal for every $U\in \unitarygroup(n)$
(hence $\|\eta(U)\|_\infty\leq 1$) and definition of $M_{n,k}$ in \eqref{eq:mnkudef}, we have
\begin{align}
\max \{\|M\|_\infty,\|L\|_\infty\}\leq \|Z\|_\infty\ .
\end{align}
Hence we conclude that 
\begin{align}
\left|\trace \left[(JM)^4 - (JL)^4\right]\right|& \leq 4\|Z\|_\infty^3 \|M-L\|_1\ . 
\end{align}
We can bound the second term on the rhs.~of Eq.~\eqref{eq:fufvdiff} in a similar manner, obtaining
\begin{align}
\left |\trace \left[(JM)^2 - (JL)^2\right]\right| &\leq  \| JM (JM-JL) \|_1 + \|  (JM-JL) JL\|_1
\leq 2 \|Z\|_\infty \|M-L\|_1\ .
\end{align}
To summarize, we get
\begin{align}
|f(U) - f(V)| &\leq (4\|Z\|_\infty^3+4\lambda^2\|Z\|_\infty)\|M-L\|_1\\
&\leq (4\|Z\|_\infty^3+4\lambda^2\|Z\|_\infty)\sqrt{2k}\|M-L\|_2 \label{eq:fuvdifferenceintermediate}
\end{align}
where we used that $\|A\|_1\leq \sqrt{2k}\|A\|_2$ if $A$ is a $2k\times 2k$-matrix.  

Since applying projections do not increase norms, \eqref{eq:rotated-M} implies in a similar way as before that
\eq{
\|M-L\|_2 &\leq \|  \eta(U) \left(Z_n\oplus Z_n^{-1} \right) \left(\eta(U)^T - \eta(V)^T\right)\| _2 + \|\left(\eta(U) - \eta(V) \right) \left(Z_n\oplus Z_n^{-1}\right) \eta(V)^T\| _2 \\
&\leq 2\|Z\|_\infty \|\eta(U) - \eta(V)\|_2 \leq 4 \|Z\|_\infty \|U - V\|_2
}
where the last inequality comes from \eqref{eq:isomorphismtosymplecticgroup}.

Combining these inequalities we conclude that 
\eq{
|f(U) - f(V)| \leq 32 \sqrt{2k} \| Z\|_\infty^4 \|U-V\|_2
}
which completes the proof. 
\end{proof}

Finally, we show that the symplectic eigenvalues of random  covariance matrices (associated with reduced density operators of random pure Gaussian states) are typically close to the average energy~$\lambda(n)$. We do so by showing that~$f(U)$ is typically small as $n \to \infty$. 
\begin{theorem}[Typical symplectic spectrum of reduced covariance matrices]\label{theorem:thermal}
Let $0 \leq \zeta$ and $0\leq \kappa$  be such that $8 \zeta + \kappa < 1$ and $\zeta + 2 \kappa <1$.
Let $\{Z_n\}_{n\in\mathbb{N}}$ be a sequence of matrices with bounded squeezing of degree~$\zeta$, and $\{k_n\}_{n\in\mathbb{N}}$ a sequence of subsystem sizes bounded of degree~$\kappa$. 
Then the symplectic eigenvalues $\{\lambda^{(n)}_j\}_{j=1}^{k_n}$ of $M_{n,k_n}$ converge in probability to the average energy~$\lambda(n)$ in the following sense.  There are universal constants $C,c>0$ such that for any $\epsilon > C n^{\zeta + 2 \kappa -1} $ we have 
\eq{
\Pr \left\{ \sum_{j=1}^k \left(\left(\lambda_j^{(n)}\right)^2 - \lambda(n)^2 \right)^2>\epsilon \right\}
\leq \exp \left( - c \epsilon^2 n^{1-8\zeta - \kappa} \right)\ .\label{eq:upperboundprlambdajn} 
}
\end{theorem}

\begin{proof}
In light of Lemma \ref{lemma:f-formula}, we can prove this using the function $f$ defined by~\eqref{equation:function}: we need to show that $\Pr(f(U)>2\epsilon)$ is upper bounded by the quantity on the rhs.~of Eq.~\eqref{eq:upperboundprlambdajn}. 
To this end, let $C>0$ be the universal constant from Lemma~\ref{lemma:average-f} and 
take $\epsilon > C n^{\zeta + 2 \kappa -1}$ so that $\E [f] < \epsilon$.
Then, by using Lemma \ref{lemma:measure-concentration} (a standard concentration of measure result for lipschitz functions on the unitary group), we have
\eq{
\Pr \left \{f(U) > 2 \epsilon \right\} 
\leq \Pr \left \{f(U) >  \epsilon +  \E \left[ f(U) \right]  \right\} 
\leq \exp \left( -  \frac{\epsilon^2 n}{12 L^2}  \right),
}
where $L$ is Lipschitz constant for $f(U)$. 
Using the bound for Lipschitz constant~$L$ from Lemma~\ref{lemma:Lipschitz}  completes the proof.
\end{proof}

\subsection{Entropy of reduced density operators\label{sec:entropyreduceddensity}}
In this section, we analyze the von Neumann entropy~$S(\tr_{n-k}\proj{\Psi(U)})$ of reduced density operators of random bipartite pure Gaussian states~$\ket{\Psi(U)}$, i.e., the entanglement entropy of~$\ket{\Psi(U)}$ with respect to a bipartition into $k$ and $n-k$~modes. Recall from Eq.~\eqref{eq:entropy} that this quantity is
determined by the symplectic eigenvalues of the covariance matrix~$M_{n,k}(U)$. 

Computing the derivative of the function~$g$ given by~\eqref{eq:entropy},  we have
$g'(\N)=\log((\N+1)/\N)$, and evaluating this at $\N=\N(\lambda)$ gives
\begin{align*}
g'(\N(\lambda))=\log \frac{\lambda+1}{\lambda-1}=:\beta(\lambda)\ .
\end{align*}
The quantity~$\beta(\lambda)$ is called the inverse temperature. Importantly, $\beta(\cdot)$ is monotonically decreasing with increasing $\lambda$, with
\begin{align}
\beta(\lambda)<2\ \textrm{ for }\lambda>2\qquad\textrm{ and }\qquad \lim_{\lambda\rightarrow\infty}\beta(\lambda)=0\ .\label{eq:betalambda}
\end{align}
From this, we get that $G'(\lambda)=g'(\N(\lambda))\N'(\lambda)$  is equal to
\begin{align}
G'(\lambda)&=\beta(\lambda)/2\ .\label{eq:gprime}
\end{align}
 In the following, we will use the quantity
\begin{align} 
\Delta(\lambda,\{\lambda_j\}_{j=1}^k):=\sqrt{\sum_{j=1}^k (\lambda^2-\lambda^2_j)^2}\label{eq:deltadistance}
\end{align}
to quantify the deviation of the symplectic spectrum~$\{\lambda_j\}_{j=1}^k$ of a $k$-mode covariance matrix from that of the covariance matrix~$\lambda I_{2k}$ associated with a thermal state.

\begin{theorem}[Typicality of entanglement entropy]\label{sec:thermalizationintermsofentropy}
Let $0 \leq \zeta$ and $0\leq \kappa$ be such that $8\zeta + 3\kappa <1$ and $\zeta + 3\kappa <1$. 
Let $\{Z_n\}_{n\in\mathbb{N}}$ be a sequence of matrices with bounded squeezing of degree~$\zeta$
and suppose that the average energy~$\lambda(n)$ is strictly and uniformly lower bounded in~$n$ by $\mu >1$.  Let $\{k_n\}_{n\in\mathbb{N}}$ be a sequence of subsystem sizes bounded of degree~$\kappa$.    Consider the distribution~$\mu_{Z_n}$ over pure Gaussian states as in Theorem~\ref{theorem:thermal}, and let $M_{n,k_n}=M_{n,k_n}(U)$  be the (random) covariance matrix of the reduced $k_n$-mode density operators. Then the $k_n$-mode entanglement entropy converges to $S(\lambda(n)I_{2k_n})=k_n G(\lambda(n))$ in the following sense. 
There are universal constants  $C,c>0$  such that for any $ \epsilon >0 $ with  
\begin{align}
C\beta(\mu) n^{(\zeta + 3 \kappa -1)/2} < \epsilon \label{eq:epsilonlowerbound}
\end{align} 
we have
\begin{align}
\Pr \left\{ \big| k_n G(\lambda(n))  -S(M_{n,k})\big|  \leq \epsilon  \right\} \geq  1-  \exp \left( - c \frac{\epsilon^4}{\beta(\mu)^4} n^{1-8\zeta - 3\kappa} \right)\ .\label{eq:mainresxy}
\end{align}
\end{theorem}
\begin{proof}
Note that the condition $\zeta + 3\kappa <1$ guarantee that  the lhs.~of~\eqref{eq:epsilonlowerbound}
vanishes in the limit $n\rightarrow\infty$, hence Eq.~\eqref{eq:epsilonlowerbound} is meaningful. Similarly, the condition~$8\zeta + 3\kappa <1$ ensures that the rhs.~of Eq.~\eqref{eq:mainresxy} goes to~$1$ as $n\rightarrow\infty$.

Observe that  (writing $\lambda=\lambda(n)$) we have 
\eq{\label{eq:mv}
\left| G(\lambda) - G(\lambda_i) \right | < \frac{\beta(\mu)}{2} \left|\lambda - \lambda_i\right| \leq \frac{\beta(\mu)}{2} \left|\lambda^2 - \lambda_i^2\right| \qquad\textrm{ for }\qquad i=1,\ldots,k\ 
}
by the mean value theorem: $\beta (\lambda)$ is positive and decreasing,
hence $\beta(\lambda)\leq \beta(\mu)$ by our assumption on~$\lambda$,
and so $G^\prime(\lambda)\leq \beta(\mu)/2$.
Note that when $\Delta(\lambda,\{\lambda_j\}_{j=1}^k)$ is small each $\lambda_i$ is close to $\lambda$.
The second equality comes from the fact that $\lambda, \lambda_i \geq 1$.
Indeed, $x-y \leq (x-y) (x+y) = x^2 - y^2 $ for $x \geq y \geq 1$. Denoting the symplectic eigenvalues of $M_{n,k}$ by $\{\lambda_j\}_{j=1}^k$, this implies 
\begin{align}
\big|k G(\lambda)-S(M_{n,k})\big|&\leq \sum_{j=1}^k \big| G(\lambda)-G(\lambda_j)\big|
\leq \frac{\beta(\mu)}{2}\sqrt{k} \Delta(\lambda,\{\lambda_j\}_{j=1}^k)\ ,
\end{align}
that is,
\begin{align}
\big|k G(\lambda)-S(M_{n,k})\big|&\leq C'\beta(\mu) n^{\kappa/2}\Delta(\lambda,\{\lambda_j\}_{j=1}^k)\label{eq:uniqual}
\end{align}
for a universal constant~$C'>0$.

Now assume that~$\epsilon'$ is such that 
\begin{align}
\epsilon'>Cn^{\zeta + 2 \kappa -1}\label{eq:pesilonprime}
\end{align} (where $C$ is the constant from Theorem~\ref{theorem:thermal}). Then~\eqref{eq:uniqual} and \eqref{eq:upperboundprlambdajn}  imply
\begin{align}
\Pr \left\{ \big| k G(\lambda)  -S(M)\big|  \leq C'\beta(\mu) n^{\kappa/2}\sqrt{\epsilon'}\right\} 
&\geq 1- \Pr \left\{\Delta(\lambda,\{\lambda_j\}_{j=1}^k)\leq \sqrt{\epsilon'} \right\}\\
&\geq 1-\exp\left(-c \epsilon'^2 n^{1-8\zeta-\kappa}\right)
\end{align}
Substituting $\epsilon:=C'\beta(\mu)n^{\kappa/2}\sqrt{\epsilon'}$, condition~\eqref{eq:pesilonprime} becomes
\begin{align}
\epsilon > C'\sqrt{C}\beta(\mu) n^{(\zeta + 3\kappa -1)/2}\ ,
\end{align}
and we get
\begin{align}
\Pr \left\{ \big| k G(\lambda)  -S(M)\big|  \leq \epsilon\right\}\geq 1-\exp\left(-\frac{c\epsilon^4}{(C')^4 \beta(\mu)^4 n^{2\kappa}}n^{1-8\zeta-\kappa}\right)\ ,
\end{align}
which completes the proof.

\end{proof}

\subsection*{Acknowledgements}
RK is supported
by the Technical University of Munich -- Institute for Advanced
Study, funded by the German Excellence Initiative and the European Union
Seventh Framework Programme under grant agreement no.~291763. He
gratefully acknowledges  support by DFG project no.~KO5430/1-1.
M.F. was financially supported by JSPS KAKENHI Grant Number~JP16K00005 
and thanks the hospitality of TUM.

\begin{appendices}
\numberwithin{equation}{section}
\renewcommand{\theequation}{\Alph{section}\arabic{equation}}

\section{Weingarten calculus}\label{sec:weingarten}
In this section we explain how to calculate expected values of polynomials of unitary matrices  along the lines of~\cite{CollinsSniady2006}.
For  the unitary group $\unitarygroup(n)$ equipped with the normalized Haar measure,
we have the following computational rule:
\begin{align}
\E_{U \in \mathcal \unitarygroup(n)} \left[ \prod_{x=1}^p u_{i_x,j_x} \prod_{y=1}^p \bar u_{i^\prime_y,j^\prime_y} \right]
= \sum_{\alpha,\beta \in S_p} \prod_{x=1}^p \delta_{i_x,i^\prime_{\alpha(x)}}  \prod_{y=1}^p \delta_{j_y,j^\prime_{\beta(y)}} \mathrm{Wg}(n, \alpha^{-1}\beta)\ .\label{eq:wgfcsum}
\end{align}
When the number of factors of $U$ and $\bar U$ are different from each other, the average vanishes. 
Let us explain the above notations: the quantity~$u_{i,j}$ is the $(i,j)$-th element of a unitary matrix $U$, $S_p$ is the permutation group of order $p$, and $\delta_{\cdot,\cdot}$ is the Kronecker delta function. The expression $\mathrm{Wg}(n, \alpha^{-1}\beta)$ is the Weingarten function, given by
\eq{
\mathrm{Wg}(n,\sigma) = \frac 1{p!^2} \sum_{\substack{\lambda \vdash p\\ l(\lambda) \leq n }} 
\frac{(\chi^\lambda (1))^2 \chi^\lambda(\sigma)}{s_{\lambda,n}(1)}\qquad\textrm{ for }\sigma\in S_p\ .
}
In this expression, the sum is over all Young tableaux~$\lambda$ with $p$ boxes and at most $n$ rows (Here $l(\lambda)$ denotes the number of rows). Furthermore, $\chi^\lambda$ is the character associated with the irreducible representation of $S_p$ labeled~$\lambda$. Finally, $s_{\lambda,n}$ is the Schur polynomial (giving the character of the associated representation of the unitary group~$\unitarygroup(n)$), 
which means that $s_{\lambda,n}(1)$ is the dimension of representation $\mathcal \unitarygroup(n)$ corresponding to a tableaux $\lambda$ (for $l(\lambda)\leq n$). 

We remark that the condition $l(\lambda) \leq n$  is often dropped in the definition of the Weingarten function -- and we will do so henceforth. As explained in~\cite{CollinsSniady2006}, this can be justified as follows. First, the condition is obviously vacuous for~$p \leq n$ (and in fact, we are only interested in the large~$n$-limit).  For $p>n$, the Weingarten function defined by summing over all Young tableaux~$\lambda$ with~$p$ boxes (instead of only those satisfying $l(\lambda)\leq n$) generally contains poles. However, these poles cancel in the sum when an expression of the form~\eqref{eq:wgfcsum} is computed.

For example, for $p=2$, we have $S_2 = \{(1)(2), (1,2) \}$ so that
\eq{
\mathrm{Wg}(n, (1)(2)) = \frac{1}{n^2-1} \eqtext{and} \mathrm{Wg}(n, (1,2)) =  \frac{-1}{n(n^2-1)}
}
By using this formula, we prove the following lemma used in the proof of Theorem~\ref{theorem:average}.
\begin{lemma}\label{lemma:calculation-average}
For any three matrices $A,B,\Pi \in \mathsf{Mat}_{n\times n}(\mathbb{C})$, we have 
\eq{
&\E_{U \in \mathcal \unitarygroup(n)} \left[UBU^*\Pi UBU^*-UAU^T\Pi \bar UAU^*\right] \\
&=   \frac{1}{n^2 -1} \left[ (\trace B)^2 \Pi  + (\trace B^2 ) (\trace\Pi ) I - (\trace A^2)  \Pi - (\trace A^2)( \trace\Pi ) I \right] \\
&+ \frac{1}{n(n^2 -1)} \left[ -(\trace B)^2 (\trace\Pi ) I - (\trace B^2 ) \Pi + (\trace A^2)( \trace\Pi) I + (\trace A^2) \Pi  \right]
}
\end{lemma}
\begin{proof}
By using the graphical calculus described in~\cite{CollinsNechita2010} (which
 amounts to using~\eqref{eq:wgfcsum}),  one obtains
\eq{
\E \left[UBU^*\Pi UBU^*\right]  = \frac{1}{n^2 -1} \left[ (\trace B)^2 \Pi  + (\trace B^2 ) (\trace\Pi ) I \right]
- \frac{1}{n(n^2 -1)} \left[ (\trace B)^2 (\trace\Pi ) I + (\trace B^2 ) \Pi  \right]  
}
Here, the first two terms correspond to $\alpha = \beta$ and the last two to $\alpha \not =\beta$.
Similarly, one has
\eq{
\E \left[UAU^T\Pi \bar UAU^*\right] = \frac{1}{n^2 -1} \left[ (\trace A^2)  \Pi + (\trace A^2)( \trace\Pi ) I \right]
- \frac{1}{n(n^2 -1)} \left[ (\trace A^2)( \trace\Pi) I + (\trace A^2) \Pi  \right]\ .
}
Here, the first two terms correspond to $\alpha = \beta$ and the last two to $\alpha \not =\beta$. This implies the claim.
\end{proof}

\section{Concentration of measure on the unitary group}
In this section, we briefly introduce a result of concentration of measure on the unitary group.
The following lemma is a special case of Corollary 17 of \cite{meckes2013},
and we explain the proof below for the reader's convenience. 
\begin{lemma}[Concentration of measure] \label{lemma:measure-concentration}
For an $L$-Lipschitz function on the unitary group
\eq{
f: \mathcal \unitarygroup(n) \to \mathbb R
}
we have the following concentration of measure phenomenon:
\eq{
\Pr \left\{ f (U) > \E  [f] + \epsilon \right\} < \exp \left( -\frac{\epsilon^2n }{12 L^2}\right)
}
for any $\epsilon >0$.
Here,  the unitary group~$\mathcal \unitarygroup(n)$ is equipped with the  Euclidean  (i.e., Hilbert-Schmidt) distance and 
$\E[\cdot]$ denotes the average over the normalized Haar measure.
\end{lemma}
\begin{proof} We summarize the idea of the proof in \cite{meckes2013}.
 Theorem 15 of \cite{meckes2013} implies that for all locally Lipschitz functions $h: \unitarygroup(n) \to \mathbb R$ we have 
\begin{align}
\mathrm{Ent} (h^2) \leq \frac{12}{n} \E \left[\left| \nabla h \right|^2 \right]\label{eq:logsobol}
\end{align}
which is expressed as ``$h$ satisfies a logarithmic Sobolev inequality with constant $\frac n 6$''. This inequality involves  two quantities defined for an arbitrary function~$g:\unitarygroup(n)\rightarrow\mathbb{R}$ which is locally Lipschitz: the quantity
\eq{
\mathrm{Ent} (g) = \E \left[g \log (g) \right] - \E(g) \log \left( \E(g)\right)   
}
and 
\eq{
\left| \nabla g \right|(U) =  \limsup_{V \to U} \frac{|g(V) - g(U)|}{d(U,V)}
}
for $U \in \unitarygroup(n)$. With~Theorem 5.3 of \cite{ledoux2001_book}, the inequality~\eqref{eq:logsobol} implies that for $\epsilon >0$
\eq{
\Pr \left\{f(U) \geq \E[f] + \epsilon\right\} \leq \exp \left( - \frac{\epsilon^2n }{12 L^2} \right)
}
This completes the proof.
\end{proof}
We refer to~\cite{ledoux2001_book}  and~\cite{MilmanSchechtman1985} for in-depth treatments of concentration of measure.

\section{Second and fourth moments\label{appendix:moments}}
\newcommand*{\lead}{\mathsf{lead}}
Here we present both the second moment (to illustrate the methodology) and the fourth moment of the matrix~$JM_{n,k}$.
These are obtained using the graphical calculus described in~\cite{CollinsNechita2010}.
 We use the \textsf{RTNI}~package~\cite{fukudanechitakoenig19} for integration over the unitary group. \textsf{MATHEMATICA} and \textsf{python} programs producing these expressions are provided separately, at
 \href{https://arxiv.org/abs/1903.04126}{arXiv:1903.04126 [quant-ph]}.

The second moment~$\E [\tr\left((JM_{n,k})^2\right) ]=\sum_{m}\alpha_m m(A,B)$ can be expressed as a linear combination of degree-$4$-monomials~$m(A,B)$ in $A$ and $B$ (cf.~Expression~\eqref{eq:secondmomentpolynomial}). Using the graphical calculus, one obtains coefficients given by Table \ref{table:2nd}.
\begin{table}
\caption{Second moments.}
\label{table:2nd}
\begin{align}
\begin{array}{cc||cc||c}
\alpha_m & \lead(\alpha_m) & m(A,B) & \lead(m(A,B)) & \lead (\alpha_m m(A,B))\\
\hline\hline
 \frac{2 k (k-n)}{n \left(n^2-1\right)} & -2 K n^{\kappa-2} & \text{Tr}[B]^2 & \lambda^2 n^2 & -2 K \lambda^2 n^{\kappa} \\
 \hline
 \frac{2 k (k+1)}{n (n+1)} & 2 K^2 n^{2 \kappa-2} & \text{Tr}[A^2] & C \lambda n^{\zeta+1} & 2 C K^2 \lambda n^{2 \kappa+\zeta-1} \\
 -\frac{2 k (k n-1)}{n \left(n^2-1\right)} & -2 K^2 n^{2 \kappa-2} & \text{Tr}[B^2] & C \lambda n^{\zeta+1} & -2 C K^2 \lambda n^{2 \kappa+\zeta-1} \\
\end{array}\\
\end{align}
\end{table}
In this table, the first (double) column gives the coefficients~$\alpha_m$ of one of terms $m\in \{(\tr B)^2,\tr(A^2),
\tr(B^2)\}$, and its leading (highest) order as a function of~$n$ when  $k=O(n^\kappa)$ and when the
asymptotic bounds~\eqref{eq:asymptoticboundv} are used.
The second (double) column specifies the monomial~$m(A,B)$ and its leading order, and the third column gives the leading order of the product $\alpha_m m(A,B)$.  We use
two constants $K, C>0$: we assume that $k\leq Kn^\kappa$ and $\|B\|_\infty\leq C n^{\zeta}$
such that~\eqref{eq:asymptoticboundv}
takes the form
\begin{align}
\Big|\tr\left[\prod_{i=1}^n X_i\right]\Big|\leq \lambda(n) C^{m-1} n^{1+(m-1)\xi}\ ,
\end{align}
where $X_i\in \{A,B\}$ for each $i=1,\ldots,n$.

We remark that the term associated with the first line in this table is special and will be treated separately as we can explicitly compute $\tr(B)^2=\lambda^2 n^2$. That is, we have
\begin{align}
\frac{2 k (k-n)}{n \left(n^2-1\right)}\cdot \tr(B)^2 &=-2k \lambda^2+O(n^{2\kappa-1})\ .
\end{align}
and thus (combined with the above table)
\begin{align}
\E [\tr\left((JM)^2\right) ]=-2k \lambda^2+O(n^{2\kappa+\zeta-1})\ .\label{eq:secondmomentsrecovered}
\end{align}
Thus we recover Corollary~\ref{cor:secondmoments}.

The fourth moment~$\E [\tr\left((JM)^4\right) ]=\sum_m \alpha_m m(A,B)$ can be expressed similarly as a  linear combination
of degree-$4$-monomials in $A$ and $B$.  The corresponding coefficients are given by Table \ref{table:4th}.
\begin{center}
\begin{table}
\caption{Forth moments.}
\label{table:4th}
\begin{align}
{\tiny
\begin{array}{cc||cc||c}
\alpha_m & \lead(\alpha_m) & m(A,B) & \lead(m(A,B)) & \lead (\alpha_m m(A,B))\\
\hline\hline
\frac{2 k \left(-5 k^3+10 n k^2-\left(6 n^2+1\right) k+n^3+n\right)}{n \left(n^6-14 n^4+49 n^2-36\right)}&  2 K n^{\kappa-4} &  \text{Tr}[B]^4 & \lambda^4 n^4 & 2 K \lambda^4 n^{\kappa}\\
\hline
 -\frac{8 k \left(\left(n^2+1\right) k^3-n \left(n^2+11\right) k^2+11 \left(n^2+1\right) k-n \left(n^2+11\right)\right)}{n \left(n^6-14 n^4+49 n^2-36\right)} & 8 K^3 n^{3 \kappa-4} & \text{Tr}[B] \text{Tr}[B^3] & C^2 \lambda^2 n^{2 (\zeta+1)} & 8 C^2 K^3 \lambda^2 n^{3 \kappa+2 \zeta-2} \\
 \frac{4 k \left(5 n k^3-2 \left(4 n^2+9\right) k^2+n \left(3 n^2+28\right) k-10 n^2\right)}{n \left(n^6-14 n^4+49 n^2-36\right)} & 12 K^2 n^{2 \kappa-4} & \text{Tr}[B]^2 \text{Tr}[B^2] & C \lambda^3 n^{\zeta+3} & 12 C K^2 \lambda^3 n^{2 \kappa+\zeta-1} \\
 \frac{2 k \left(\left(n^3+n\right) k^3-20 n^2 k^2+5 n \left(n^2+13\right) k-4 \left(4 n^2+9\right)\right)}{n \left(n^6-14 n^4+49 n^2-36\right)} & 2 K^4 n^{4 \kappa-4} & \text{Tr}[B^4] & C^3 \lambda n^{3 \zeta+1} & 2 C^3 K^4 \lambda n^{4 \kappa+3 \zeta-3} \\
 \frac{2 k \left(\left(3-2 n^2\right) k^3+2 n \left(n^2+6\right) k^2-\left(16 n^2+21\right) k+n \left(n^2+21\right)\right)}{n \left(n^6-14 n^4+49 n^2-36\right)} & 4 K^3 n^{3 \kappa-4} & \text{Tr}[B^2]^2 & C^2 \lambda^2 n^{2 (\zeta+1)} & 4 C^2 K^3 \lambda^2 n^{3 \kappa+2 \zeta-2} \\
 -\frac{4 k \left((3 n+4) k^3-2 n (2 n+1) k^2+\left(n^3-3 n^2-n-4\right) k+n \left(n^2+n+4\right)\right)}{(n-2) (n-1) n^2 (n+1) (n+2) (n+3)} & -4 K^2 n^{2 \kappa-4} & \text{Tr}[B]^2 \text{Tr}[A^2] & C \lambda^3 n^{\zeta+3} & -4 C K^2 \lambda^3 n^{2 \kappa+\zeta-1} \\
 -\frac{4 k (k+1) \left((n+1) k^2-\left(n^2+1\right) k-(n-1) n\right)}{(n-1) n^2 (n+1) (n+2) (n+3)} & 4 K^3 n^{3 \kappa-4} & \text{Tr}[A^2]^2 & C^2 \lambda^2 n^{2 (\zeta+1)} & 4 C^2 K^3 \lambda^2 n^{3 \kappa+2 \zeta-2} \\
 \frac{2 k (k+1) \left(\left(n^2+n+2\right) k^2+\left(3 n^2-5 n-2\right) k+4 (n-1) n\right)}{(n-1) n^2 (n+1) (n+2) (n+3)} & 2 K^4 n^{4 \kappa-4} & \text{Tr}[A^4] & C^3 \lambda n^{3 \zeta+1} & 2 C^3 K^4 \lambda n^{4 \kappa+3 \zeta-3} \\
 \frac{4 k (k+1) \left(\left(n^2+5 n+4\right) k^2-\left(n^2+n+4\right) k-2 n (n+1)\right)}{(n-1) n^2 (n+1) (n+2) (n+3)} & 4 K^4 n^{4 \kappa-4} & \text{Tr}[(AB)^2] & C^3 \lambda n^{3 \zeta+1} & 4 C^3 K^4 \lambda n^{4 \kappa+3 \zeta-3} \\
 -\frac{8 k \left(\left(n^3+2 n^2-n-4\right) k^3+n \left(n^2-5 n-4\right) k^2+\left(n^3-8 n^2+5 n+4\right) k+n \left(n^2-n+8\right)\right)}{(n-2) (n-1) n^2 (n+1) (n+2) (n+3)} & -8 K^4 n^{4 \kappa-4} & \text{Tr}[A^2B^2] & C^3 \lambda n^{3 \zeta+1} & -8 C^3 K^4 \lambda n^{4 \kappa+3 \zeta-3} \\
 \frac{4 k \left(\left(2 n^2+3 n-4\right) k^3-2 n \left(n^2+n-1\right) k^2+\left(-n^3+n^2-5 n+4\right) k+n \left(n^2+5 n-4\right)\right)}{(n-2) (n-1) n^2 (n+1) (n+2) (n+3)} & -8 K^3 n^{3 \kappa-4} & \text{Tr}[A^2] \text{Tr}[B^2] & C^2 \lambda^2 n^{2 (\zeta+1)} & -8 C^2 K^3 \lambda^2 n^{3 \kappa+2 \zeta-2} \\
 \frac{8 k \left(\left(n^2+n+4\right) k^3+n \left(-n^2+n-8\right) k^2-\left(2 n^3-5 n^2+5 n+4\right) k+n \left(-n^2+5 n+4\right)\right)}{(n-2) (n-1) n^2 (n+1) (n+2) (n+3)} & -8 K^3 n^{3 \kappa-4} & \text{Tr}[B] \text{Tr}[A^2B] & C^2 \lambda^2 n^{2 (\zeta+1)} & -8 C^2 K^3 \lambda^2 n^{3 \kappa+2 \zeta-2} \\
\end{array}}
\end{align}
\end{table}
\end{center}
Again, we treat the first line in this table differently: we have
\begin{align}
\frac{2 k \left(-5 k^3+10 n k^2-\left(6 n^2+1\right) k+n^3+n\right)}{n \left(n^6-14 n^4+49 n^2-36\right)}
\cdot \tr(B)^4&=2k \lambda^4
+O(n^{2\kappa-1})\ .
\end{align}
With the above table and the fact that, by assuming $\kappa+\zeta\leq 1$ as before,
\begin{align}
\max \{2\kappa-1,2 \kappa+\zeta-1,3 \kappa+2 \zeta-2,4 \kappa+3\zeta-3)\}&= 2\kappa+\zeta-1 .
\end{align}
we conclude that 
\begin{align}
\E [\tr\left((JM)^4\right) ]&=2k \lambda^4+ O(n^{2\kappa+\zeta-1}).
\end{align}
This proves Lemma~\ref{lem:fourthmomentsexplicit}.

\end{appendices}


\def\cprime{$'$}

\end{document}